\documentclass[twoside,leqno,twocolumn]{article}
\usepackage[nocopyright,siamsections,siamtitle,smallbib]{jeffeproc}

\usepackage{jeffe,graphicx,hyperref}
\usepackage[utf8]{inputenc}
\usepackage[margin=1in]{geometry}
\usepackage{marvosym}
\usepackage{picins}
\usepackage{smallbib}
\usepackage{cite}
\usepackage{verbatim}

\usepackage{microtype}
\usepackage[charter]{mathdesign}

\usepackage[mathcal]{euscript}
\usepackage{stmaryrd}

\def\arcto{\mathord\shortrightarrow}

\def\rev{\mathit{rev}}
\def\Z{\mathbb{Z}}
\def\Real{\mathbb{R}}

\def\snip{\mathbin{\raisebox{0.15ex}{\rotatebox[origin=c]{60}{\Rightscissors}\!}}}

\let\path p

\newtheorem{theorem}{Theorem}[section]
\newtheorem{corollary}[theorem]{Corollary}
\newtheorem{lemma}[theorem]{Lemma}

\begin{document}

\title{Shortest Non-trivial Cycles in Directed and Undirected Surface Graphs}

\author{
  Kyle Fox%
  \thanks{Department of Computer Science,
      University of Illinois, Urbana-Champaign;
      \url{kylefox2@illinois.edu}. Portions of this work were done while the author was
				visiting Google, Inc.  Research
				supported in part by the Department of Energy Office
				of Science Graduate Fellowship Program (DOE SCGF),
				made possible in part by the American Recovery and
				Reinvestment Act of 2009, administered by ORISE-ORAU
				under contract no. DE-AC05-06OR23100.}
}


\maketitle

\begin{abstract}
  Let~$G$ be a graph embedded on a surface of genus~$g$
  with~$b$ boundary cycles.
  We describe algorithms to compute multiple types of non-trivial cycles in~$G$, using different techniques depending on whether or not $G$ is an undirected graph.
  If $G$ is undirected, then we give an algorithm to compute a shortest non-separating cycle in $2^{O(g)} n \log \log n$ time. Similar algorithms are given to compute a shortest non-contractible or non-null-homologous cycle in $2^{O(g + b)} n \log \log n$ time. Our algorithms for undirected $G$ combine an algorithm of Kutz with known techniques for efficiently enumerating homotopy classes of curves that may be shortest non-trivial cycles.
  
  Our main technical contributions in this work arise from assuming $G$ is a \emph{directed} graph with possibly asymmetric edge weights.
  For this case, we give an algorithm to compute a shortest non-contractible cycle in~$G$ in ${O((g^3 + gb)n \log n)}$ time.
  In order to achieve this time bound, we use a restriction
  of the infinite cyclic cover that may be useful in other contexts.
  We also describe an algorithm to compute a shortest
  non-null-homologous cycle in~$G$ in ${O((g^2+gb) n \log n)}$ time,
  extending a known algorithm of Erickson
  to compute a shortest non-separating cycle.
  In both the undirected and directed cases, our algorithms improve the
  best time bounds known for many values
  of~$g$ and~$b$.
\end{abstract}

\noindent

\thispagestyle{empty}
\setcounter{page}{0}



\section{Introduction}
There is a long line of work on computing shortest non-trivial cycles in surface
embedded graphs.
Cabello and Mohar~\cite{cm-fsnsn-07} claim that finding short
non-trivial cycles is arguably one of the most natural problems for graphs
embedded on a surface. Additionally, finding these cycles has
many benefits both for theoretical combinatorial
problems~\cite{mt-gs-01,dhm-aacd-07,km-gmiap-08,bdt-ptass-08}
and more practical applications in areas such as
graphics and graph drawing
\cite{eh-ocsd-04,bls-rrgh-09,is-pebgg-07,kr-ccnlt-07,gw-tnr-01,whds-reti-04}.

\begin{figure*}[t]
  \centering
  \includegraphics[height=1.5in]{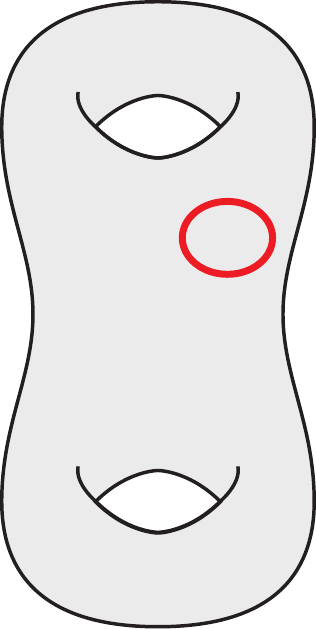}
  \qquad\qquad\qquad
  \includegraphics[height=1.5in]{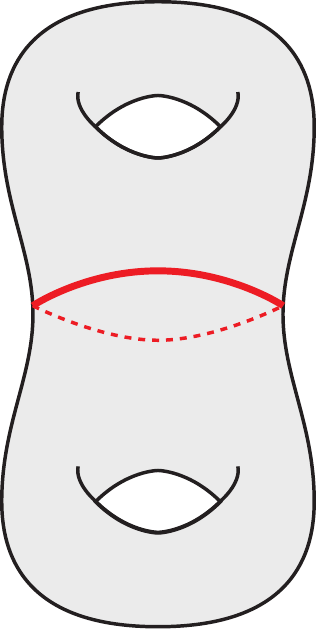}
  \qquad\quad\qquad
  \includegraphics[height=1.5in]{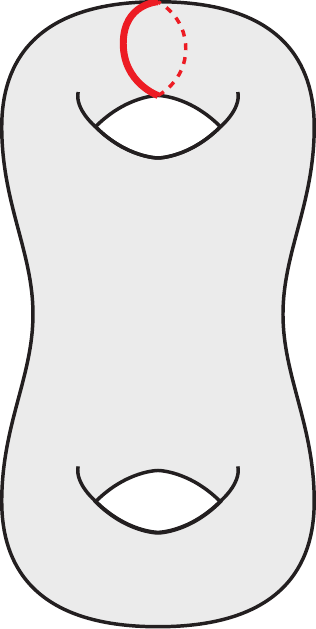}
  \caption{Left: A contractible cycle on~$\Sigma$. Center: A non-contractible but separating cycle on~$\Sigma$. Right: A non-contractible and non-separating cycle on~$\Sigma$.}
  \label{fig:cycle-types}
\end{figure*}

Researchers have focused primarily on finding short non-contractible
and non-separating cycles.
Consider a graph~$G$ embedded on a surface~$\Sigma$.
Informally,
a cycle in~$G$ is
\EMPH{contractible} if its image on~$\Sigma$ can be continuously deformed
to a single point. The cycle is \EMPH{separating} if removing its image
from~$\Sigma$ disconnects~$\Sigma$. Every non-separating cycle 
is non-contractible, but there may be non-contractible cycles that are
separating. See Figure~\ref{fig:cycle-types} for examples.

The history of non-trivial cycles in undirected graphs goes back several years
to a result of Itai and Shiloach~\cite{is-mfpn-79}. They give an~$O(n^2 \log n)$
time algorithm to find a shortest non-trivial cycle in an annulus as a
subroutine for computing minimum $s,t$-cuts in planar graphs. Their result
has seen several improvements, most recently by
Italiano \etal~\cite{r-mstcp-83,f-faspp-87,insw-iamcmf-11}.

Thomassen~\cite{t-egnsn-90} gave the first efficient algorithm for computing
non-trivial cycles on surfaces with arbitrary genus. His algorithm runs
in~$O(n^3)$ time and relies on a property of certain families of cycles
known as the \EMPH{3-path condition};
see also Mohar and Thomassen~\cite[Chapter 4]{mt-gs-01}. Erickson and
Har-Peled~\cite{eh-ocsd-04} gave an~$O(n^2 \log n)$ time algorithm, which
remains the fastest known for graphs of arbitrary genus.
Cabello and Mohar~\cite{cm-fsnsn-07} gave the first results parameterized
by genus, and Kutz~\cite{k-csnco-06} showed it is possible to to find
short non-trivial cycles in time near-linear in the number of vertices
if we allow an exponential dependence on the genus.
Kutz's algorithm requires searching~$g^{O(g)}$ subsets of the
\emph{universal cover}.
Cabello, Chambers,
and Erickson~\cite{cc-msspg-07,cce-msspe-12} later showed the near-linear
time dependence is possible with only a polynomial dependence on the genus
by avoiding use of the universal cover.
The current best
running time in terms of the number of vertices
is~$g^{O(g)}n \log \log n$ due to a modification to Kutz's algorithm by Italiano
\etal~\cite{insw-iamcmf-11}.
For other results
related to finding interesting cycles on surfaces,
see~\cite{c-fscss-10,ccl-osaew-10,ccl-fctpe-10,cdem-fotc-10,ccelw-scsih-08,cen-mcshc-09,ew-csec-10,efn-gmcse-12}.

Unfortunately, all of the above results rely on properties that exist only
in \emph{undirected} graphs; shortest paths intersect at most once (assuming
uniqueness), and the reversal of any shortest path is a shortest path.
Due to the difficulty in avoiding these assumptions, there are few results
for finding shortest non-trivial cycles in directed surface graphs, and all
of these results are relatively recent.
Befittingly, the short history of these results appears to coincide nicely
with the history given above for undirected graphs.

Janiga and Koubek~\cite{jk-mcdpn-92} gave the first near-linear time algorithm
for computing a shortest non-trivial cycle in a directed graph embedded on
an annulus as an attempt to find minimum $s,t$-cuts in planar
graphs\footnote{Unfortunately, their minimum cut algorithm has a subtle
error~\cite{kn-mcupg-11} which may lead to an incorrect result when the
minimum $t,s$-cut is smaller than the minimum $s,t$-cut.}. Their result
can also be achieved using recent maximum flow algorithms for
planar graphs~\cite{bk-amfdp-09,e-mfpsp-10,w-mstfp-97}.

Cabello, Colin de Verdi\`ere, and Lazarus~\cite{ccl-fsncd-10}
gave the first efficient algorithms for
computing shortest non-trivial cycles in directed surface graphs of arbitrary
genus. Their algorithms run in~$O(n^2 \log n)$ time
and~$O(\sqrt{g}n^{3/2}\log n)$ time, and rely on a variant of
the 3-path condition and balanced separators, respectively.
Erickson and Nayyeri~\cite{en-mcsnc-11} gave a~$2^{O(g)}n \log n$ time algorithm
for computing the
shortest \emph{non-separating} cycle that relies on computing
the shortest cycle in each of~$2^{O(g)}$ \emph{homology classes}.
The latest results for these problems are two algorithms of
Erickson~\cite{e-sncds-11}. The first algorithm
computes shortest non-separating cycles
in~$O(g^2 n \log n)$ time by computing shortest paths in several
copies of a linear sized covering space.
The second algorithm computes shortest non-contractible
cycles (which may be separating)
in~$g^{O(g)} n \log n$ time in a manner similar to Kutz's
algorithm~\cite{k-csnco-06},
by lifting the graph to a finite
(but large) subset of the universal cover.

\subsection{Our results}

In both the undirected and directed graph settings, researchers presented
near-quadratic time algorithms
for computing non-separating and non-contractible
cycles, and others supplemented them with algorithms with
exponential dependence in the genus, but near-linear dependence in the
complexity of the embedded graph. Similar trends appear in the computation
of maximum flows and minimum cuts in surface embedded
graphs
\cite{st-dsdt-83,gt-namfp-88,cen-hfcc-09,cen-mcshc-09,en-mcsnc-11,efn-gmcse-12}.
For the problems mentioned in this paragraph, we ideally would like
algorithms with a near-linear dependency on graph complexity but
only a polynomial dependence on genus.
Of course, we are also interested in pushing down the dependence on graph complexity even if it means sacrificing a bit in the genus dependency when $g$ is sufficiently small.

Our first result is improved algorithms for computing non-trivial cycles in \emph{undirected} surface graphs. Our algorithms run in $2^{O(g)} n \log \log n$ time and can be used to find shortest non-separating, non-contractible, or \emph{non-null-homologous} cycles. Informally, a \EMPH{non-null-homologous}
cycle is one that is either non-separating or separates a pair of boundary
cycles.
These algorithms improve the running times achieved by Italiano \etal~\cite{insw-iamcmf-11} for finding shortest non-separating and non-contractible cycles and show that it is possible to take advantage of the universal cover as in Kutz's algorithm in order to minimize the dependency on $n$, without searching a super-exponential in $g$ number of subsets of the covering space. For surfaces with $b$ boundary cycles, the shortest non-contractible and non-null-homologous cycle algorithms run in time $2^{O(g + b)}n \log \log n$, while the shortest non-separating cycle algorithm continues to run in $2^{O(g)}n \log \log n$ time. The main idea behind these algorithms is to construct fewer subsets of the universal cover by only constructing subsets corresponding to certain weighted triangulations of a \emph{dualized polygonal schema} as in~\cite{ccelw-scsih-08,cen-mcshc-09}. These algorithms
are described in Section~\ref{sec:undirected}.

Next, we sketch an
algorithm to compute a shortest non-null-homologous cycle in a \emph{directed} surface graph
in~$O((g^2+gb)n \log n)$ time.
This algorithm is actually a straightforward extension to
Erickson's algorithm for computing shortest non-separating
cycles~\cite{e-sncds-11}, but we must work out some non-trivial details
for the sake of completeness.
The key change to Erickson's algorithm is that we compute shortest paths
in an additional~$O(b)$ copies of a covering space defined using shortest paths
between boundary cycles.
These additional computations will find a shortest non-null-homologous cycle
if all shortest non-null-homologous cycles are separating.
This algorithm is given in Section~\ref{sec:non-null-homologous}.
Along with being an interesting result in its own right, we use this algorithm
as a subroutine for our primary result described below.

Our final, primary, and most technically interesting result is an~$O(g^3 n \log n)$ time algorithm for computing
shortest non-contractible cycles in directed surface graphs, improving the result of
Erickson~\cite{e-sncds-11} for all positive~$g$ and showing it is possible
to have near-linear dependency in graph complexity without suffering an
exponential dependency on genus.
On a surface
with~$b$ boundary cycles, our algorithm runs
in~$O((g^3 + gb) n \log n)$ time. In order to achieve this
running time, we choose to forgo using a subset of the universal cover in
favor of subsets of a different
covering space known as the infinite cyclic cover.
If any shortest non-contractible cycle is non-separating, then the
algorithm of Erickson~\cite{e-sncds-11} will find such a cycle
in~$O(g^2 n \log n)$ time. On the other hand, if any shortest
non-contractible cycle is separating, then it will lift to a
non-null-homologous cycle in the subset of the infinite cyclic
cover if we lift the graph to the covering space
in the correct way.
Our description of the infinite cyclic cover and its properties appears in Sections~\ref{sec:cover} and~\ref{sec:lifting}.
The algorithm is given in section~\ref{sec:non-contractible-algorithm}.

We give preliminary material useful to understanding each of these results in Section~\ref{sec:prelims}. However,
each of the results mentioned above is described in its relevant section or sections independently of the other results.
We believe that all the algorithms mentioned above are of practical and technical interest
in their own right. Additionally, the techniques used, particularly for the directed graph algorithms,
may be useful in other
contexts such as more efficiently computing minimum cuts or maximum
flows in surface embedded graphs.

\section{Preliminaries}
\label{sec:prelims}

We begin by recalling several useful definitions related to surface-embedded graphs.  For further background, we refer the reader to Gross and Tucker \cite{gt-tgt-01} or Mohar and Thomassen~\cite{mt-gs-01} for topological graph theory, and to Hatcher~\cite{h-at-02} or Stillwell~\cite{s-ctcgt-93} for surface topology and homology.
We adopt the presentation
of our terminology and notation directly from previous works~\cite{cen-mcshc-09,en-mcsnc-11,en-sncwp-11,e-sncds-11,efn-gmcse-12}.

\subsection{Surfaces and Curves}

A \EMPH{surface} (more formally, a \emph{2-manifold with boundary}) is a compact Hausdorff space in which every point has an open neighborhood homeomorphic to either the plane $\Real^2$ or a closed halfplane $\set{(x,y)\in \Real^2\mid x\ge 0}$.  The points with halfplane neighborhoods make up the \EMPH{boundary} of the surface; every component of the boundary is homeomorphic to a circle.
A surface is \EMPH{non-orientable} if it contains a subset homeomorphic to
the M\"obius band, and \EMPH{orientable} otherwise. For this paper, we consider only compact, connected, orientable surfaces. 

A \EMPH{path} in a surface $\Sigma$ is a continuous function $p\colon [0,1]\to\Sigma$.
A \EMPH{loop} is a path whose endpoints~$p(0)$ and~$p(1)$ coincide;
we refer to this common endpoint as the \EMPH{basepoint} of the loop.
An \EMPH{arc} is a path internally disjoint from the boundary of~$\Sigma$
whose endpoints lie on the boundary of $\Sigma$.
A \EMPH{cycle} is a continuous function $\gamma\colon S^1\to\Sigma$;
the only difference between a cycle and a loop is that a loop has a
distinguished basepoint.
We say a loop~$\ell$ and a cycle~$\gamma$ are \EMPH{equivalent} if, for some
real number~$\delta$, we have~$\ell(t) = \gamma(t + \delta)$ for
all~$t \in [0,1]$.
We collectively refer to paths, loops, arcs, and cycles as \EMPH{curves}.
A curve is \EMPH{simple} if it is injective; we usually do not distinguish between simple curves and their images in $\Sigma$.
A simple curve~$p$ is \EMPH{separating} if~$\Sigma \setminus p$ is disconnected.

The \EMPH{reversal}~$\rev(p)$ of a path~$p$ is defined by
setting~$\rev(p)(t) = p(1-t)$. The \EMPH{concatenation}~$p \cdot q$ of two
paths~$p$ and~$q$ with~$p(1)=q(0)$ is the path created by
setting~$(p\cdot q)(t) = p(2t)$ for all~$t \leq 1/2$
and~$(p\cdot q)(t) = q(2t-1)$ for all~$t \geq 1/2$. Finally, let~$p[x,y]$
denote the subpath of a path~$p$ from point~$x$ to point~$y$.

The \EMPH{genus} of a surface $\Sigma$ is the maximum number of disjoint simple cycles in $\Sigma$ whose complement is connected.
 Up to homeomorphism,
there is exactly one such surface with any genus $g\ge 0$ and any number of
boundary cycles $b\ge 0$; the \EMPH{Euler characteristic}~$\chi$ of this
surface is~$\chi := 2 - 2g - b$.

\subsection{Graph Embeddings}

An \EMPH{embedding} of an undirected graph $G$ on a surface $\Sigma$ maps vertices to distinct points and edges to simple, interior-disjoint paths.  The \EMPH{faces} of the embedding are maximal connected subsets of $\Sigma$ that are disjoint from the image of the graph.  An embedding is \EMPH{cellular} if each of its faces is homeomorphic to the plane; in particular, in any cellular embedding, each component of the boundary of $\Sigma$ must be covered by a cycle of edges in $G$.  Euler's formula implies that any cellularly embedded graph with $n$ vertices, $m$ edges, and $f$ faces lies on a surface with Euler characteristic $\chi = n-m+f$, which implies that $m = O(n+g)$ and $f=O(n+g)$
if the graph is simple.
We consider only such
cellular embeddings of genus $g=O(\sqrt{n})$, so that the overall complexity of the embedding is $O(n)$.

Any cellular embedding in an orientable surface can be encoded combinatorially
by a \EMPH{rotation system}, which records the counterclockwise order of edges
incident to each vertex.
Two paths or cycles in a combinatorial surface \EMPH{cross} if no continuous infinitesimal perturbation makes them disjoint; if such a perturbation exists, then the paths are \EMPH{non-crossing}.

We redundantly use the term \EMPH{arc} to refer to a walk in the graph whose endpoints are boundary vertices.  Likewise, we use the term \EMPH{cycle} to refer to a closed walk in the graph. \EMPH{Cutting} a combinatorial surface along a cycle or  arc modifies both the surface and the embedded graph.  For any combinatorial surface $S = (\Sigma, G)$ and any simple cycle or arc $\gamma$ in~$G$, we define a new combinatorial surface \EMPH{$S \snip \gamma$} by taking the topological closure of $\Sigma \backslash \gamma$ as the new underlying surface; the new embedded graph contains two copies of each vertex and edge of $\gamma$, each bordering a new boundary.

Any undirected graph~$G$ embedded on a surface~$\Sigma$ without boundary has a
\EMPH{dual graph}~$G^*$, which has a vertex~$f^*$ for each face~$f$ of~$G$,
and an edge~$e^*$ for each edge~$e$ in~$G$ joining the vertices dual to the
faces of~$G$ that~$e$ separates. The dual graph~$G^*$ has a natural cellular
embedding in~$\Sigma$, whose faces corresponds to the vertices of~$G$.
For any subgraph~$F = (U, D)$ of~$G=(V,E)$, we write~$G \setminus F$
to denote the edge-complement~$(V,E \setminus D)$. We also abuse notation
by writing~$F^*$ to denote the subgraph of~$G^*$ corresponding to any
subgraph~$F$ of~$G$.

A \EMPH{tree-cotree decomposition}~$(T,L,C)$ of an undirected
graph~$G$ embedded on a surface without boundary 
is a partition of the edges into three disjoint subsets;
a spanning tree~$T$ of~$G$, a spanning cotree~$C$ (the dual of a spanning
tree~$C^*$ of~$G^*$), and leftover edges~$L=G \setminus (T \cup C)$.
Euler's formula implies that in any tree-cotree decomposition, the set~$L$
contains exactly~$2g$ edges~\cite{e-dgteg-03}. The definitions for dual
graphs and tree-cotree decompositions given above
extend to surfaces with boundary, but
we do not require these extensions in this paper.

For some of the problems we consider, the input is actually a \emph{directed}
edge-weighted graph~$G$ with a cellular embedding on some surface. We use the
notation~$u \arcto v$ to denote the directed edge from vertex~$u$ to vertex~$v$.
Without loss of generality, we consider only \EMPH{symmetric} directed graphs,
in which the reversal~$v \arcto u$ of any edge~$u \arcto v$ is another edge,
possibly with infinite weight. We also assume that in the cellular embedding,
the images of any edge in~$G$ and its reversal coincide (but with opposite
orientations).
Thus, like Cabello \etal~\cite{ccl-fsncd-10} and Erickson~\cite{e-sncds-11},
we implicitly model directed graphs as \emph{undirected graphs with asymmetric
edge weights}. Duality can be extended to directed graphs~\cite{cen-hfcc-09},
but the results in this paper do not require this extension.

Let~$p = v_0 \arcto v_1 \arcto \dots \arcto v_k$ be a simple directed cycle
or arc in an
embedded graph~$G$. We say an edge~$u \arcto v_i$ \EMPH{enters~$p$
from the left}
(resp. right) if the vertices~$v_{i-1}$,~$u$, and~$v_{i+1}$ (module~$k$
in the case of a cycle) are ordered
clockwise (resp. counterclockwise) around~$v_i$, according to the
embedding's rotation system. An edge~$v_i \arcto u$ \EMPH{leaves~$p$ from the
left} (resp. right) if its reversal~$u \arcto v_i$ enters~$p$ from the
left (resp. right).
If~$p$ is an arc,
the above definitions require that~$0 < i < k$ and that~$u$
is not a vertex in~$p$. Recall an arc's endpoints lie on boundary cycles.
Let~$t_0 v_0$ and~$v_0 w_0$ be the boundary edges
incident to~$v_0$ with vertices~$t_0$, $v_1$, and $w_0$ appearing in clockwise
order around~$v_0$. We say~$t_0 \arcto v_0$ enters~$p$ from the left.
We say~$w_0 \arcto v_0$ enters~$p$
from the right. Similarly,
if~$t_k v_k$ and~$v_k w_k$ are boundary edges incident to~$v_k$ with
vertices~$t_k$, $w_k$, and~$v_{k-1}$ appearing in clockwise order around~$v_k$,
we say~$t_k \arcto v_k$ enters~$p$ from the left and~$w_k \arcto v_k$
enters~$p$ from the right. Finally, we treat~$t_0$ as~$v_{-1}$ and~$t_k$
as~$v_{k+1}$ to define entering from the left (resp. right) for any
other edges~$u \arcto v_0$ or~$u \arcto v_k$ where~$u$ does not appear in~$p$.

To simplify our presentation and analysis, we assume that any two
vertices~$x$ and~$y$ in~$G$ are connected by a unique shortest directed
path, denoted~$\sigma(x,y)$. The Isolation Lemma~\cite{mvv-memi-87} implies
that this assumption can be enforced (with high probability) by perturbing
the edge weights with random infinitesimal values~\cite{eh-ocsd-04}.

Our algorithms rely on a result by Cabello \etal~\cite{cc-msspg-07,cce-msspe-12} which generalizes
a result of Klein~\cite{k-msspp-05} for planar graphs.
\begin{lemma}[Cabello \etal~\cite{cc-msspg-07,cce-msspe-12}]
  \label{lem:mssp}
  Let~$G$ be a directed graph with non-negative edge weights, cellularly
  embedded on a surface~$\Sigma$ of genus~$g$, and let~$f$ be an arbitrary
  face of~$G$. We can preprocess~$G$ in~$O(gn \log n)$ time\footnote{The
  published version of this algorithm~\cite{cc-msspg-07} has a weaker time
  bound of~$O(g^2n \log n)$. Using the published
  version increases the running time
  of our algorithms for directed graphs by a factor of~$g$.} and~$O(n)$ space, so that the length
  of any shortest path from any vertex incident to~$f$ to any other vertex
  can be retrieved in~$O(\log n)$ time.
\end{lemma}

\subsection{Homotopy and Homology}

Two paths~$p$ and~$q$ in $\Sigma$ are \EMPH{homotopic} if one can be
continuously deformed into the other without changing their endpoints.
More formally, a \EMPH{homotopy} between~$p$ and~$q$ is a
continuous map $h\colon {[0,1]\times [0,1] \to \Sigma}$ such that $h(0,\cdot) = p$, $h(1,\cdot) = q$, $h(\cdot, 0)=p(0)=q(0)$, and $h(\cdot,1)=p(1)=q(1)$.
Homotopy defines an equivalence relation over the set of paths with any
fixed pair of endpoints. The set of homotopy classes of loops in~$\Sigma$
with basepoint~$x_0$ defines a group~$\pi_i(\Sigma,x_0)$ under concatenation,
called the \EMPH{fundamental group} of~$\Sigma$. (For all basepoints~$x_0$
and~$x_1$, the groups~$\pi_i(\Sigma,x_0)$ and~$\pi_i(\Sigma,x_1)$ are
isomorphic.) A cycle is \EMPH{contractible} if it is homotopic to a constant
map.

Homology is a coarser equivalence relation than homotopy, with nicer algebraic properties. Like several earlier papers
\cite{cf-qhc2-07,cen-mcshc-09, en-mcsnc-11,e-sncds-11,efn-gmcse-12},
we consider only one-dimensional cellular homology with coefficients in the finite field $\Z_2$.

Fix a cellular embedding of an undirected graph $G$ on a surface $\Sigma$
with genus $g$ and~$b$ boundary cycles.
An \EMPH{even subgraph} is a subgraph of $G$ in which every node has even degree, or equivalently, the union of edge-disjoint cycles.  An even subgraph is \EMPH{null-homologous} if it is the boundary of the closure of the union of a subset of faces of $G$.
Two even subgraphs $\eta$ and $\eta'$ are \EMPH{homologous}, or in the same \EMPH{homology class}, if their symmetric difference
$\eta\oplus\eta'$ is null-homologous.
The set of all homology classes of even subgraphs defines the
\EMPH{first homology group} of~$\Sigma$, which is isomorphic to the finite
vector space~$(\Z_2)^{2g + \max\Set{b - 1,0}}$.
If~$b \leq 1$, then a simple cycle~$\gamma$ is separating
if and only if it is null-homologous; however, when~$b>1$, some separating
cycles are not null-homologous.

\subsection{Covering spaces}

A continuous map~$\pi : \Sigma' \to \Sigma$ between two surfaces is called a
\EMPH{covering map} if each point~$x \in \Sigma$ lies in an open
neighborhood~$U$ such that (1)~$\pi^{-1}(U)$ is a countable union of disjoint
open sets~$U_1 \cup U_2 \cup \cdots$ and (2) for each~$i$, the
restriction~$\pi |_{U_i} : U_i \to U$ is a homeomorphism. If there is a covering
map~$\pi$ from~$\Sigma'$ to~$\Sigma$, we call~$\Sigma'$ a \EMPH{covering space}
of~$\Sigma$. The \EMPH{universal cover}~$\tilde{\Sigma}$ is the unique
simply-connected covering space of~$\Sigma$ (up to homeomorphism).
The universal cover is so named because it covers every path-connected
covering space of~$\Sigma$.

For any path~$p : [0,1] \to \Sigma$ such that~$\pi(x') = p(0)$ for some
point~$x' \in \Sigma'$, there is a unique path~$p'$ in~$\Sigma'$, called
a \EMPH{lift} of~$p$, such that~$p'(0) = x'$ and~$\pi \circ p' = p$. We also
say that~$p$ \EMPH{lifts} to~$p'$. Conversely, for any path~$p'$ in~$\Sigma'$,
the path~$\pi \circ p'$ is called a \EMPH{projection} of~$p'$.

We define a lift of a cycle~$\gamma : S^1 \to \Sigma$ to be the infinite
path~$\gamma' : \R \to \Sigma'$ such that
${\pi(\gamma'(t)) = \gamma(t \mod 1)}$
for all real~$t$. We call the path obtained by restricting~$\gamma'$ to any
unit interval a \EMPH{single-period lift} of~$\gamma$; equivalently, a
single-period lift of~$\gamma$ is a lift of any loop equivalent to~$\gamma$.
We informally say that a cycle is the~\EMPH{projection} of any of its
single-period lifts.


\section{Non-trivial Cycles in Undirected Graphs}
\label{sec:undirected}

Let~$G$ be an undirected graph with non-negative edge weights, cellularly embedded on an orientable surface~$\Sigma$ of genus~$g$. 
We sketch an algorithm to compute a shortest non-separating, non-contractible, or non-null-homologous cycle in~$G$. We assume the surface has no boundary, and consider the case with boundary at the end of this section.
Recall any shortest non-null-homologous cycle is a shortest non-separating cycle in a surface without boundary.

We begin by reviewing Kutz's~\cite{k-csnco-06} algorithm for computing shortest non-trivial cycles. He begins by computing a \emph{greedy system of loops}~$\Lambda = \Set{\lambda_1,\lambda_2,\dots,\lambda_{2g}}$ using a construction of Erickson and Whittlesey~\cite{ew-gohhg-05}. The construction can be performed in $O(gn)$ time using our assumption that $g = O(\sqrt{n})$~\cite{k-csnco-06}. The surface $D = \Sigma \setminus \Lambda$ is a topological disk with each loop $\lambda_i \in \Lambda$ appearing twice upon its boundary.  See Figure~\ref{fig:triangulations}. Kutz argues that there exists some shortest non-trivial cycle $\gamma$ that meets three criteria: (1) $\gamma$ crosses each loop $\lambda_i$ at most twice~\cite[Lemma 1]{k-csnco-06}; (2) the crossing sequence of $\gamma$ with regards to the loops contains no \emph{curls}; there is never any instance where $\gamma$ crosses a loop $\lambda_i$ from left-to-right (right-to-left) only to immediately cross again right-to-left (left-to-right)~\cite[Lemma 3]{k-csnco-06}; and (3) $\gamma$ is simple. 
Given a cycle $\gamma$, there exists a sequence of crossings between $\gamma$ and the loops of $\Lambda$.
Kutz uses the above observations to find shortest cycles corresponding to $g^{O(g)}$ crossing sequences of length $O(g)$ where at least one of the crossing sequences corresponds to a shortest non-trivial cycle. For each crossing sequence~$X$, he describes how to determine if a cycle corresponding to $X$ meets the criteria above and, if so, how to find a shortest cycle corresponding to $X$  in $O(g n \log n)$ time using an algorithm of Colin de Verdi\'ere and Erickson~\cite{ce-tnpcs-10}. Italiano \etal~\cite{insw-iamcmf-11} later improved the running time of Colin de Verdi\'ere and Erickson's algorithm  to $O(n \log \log n)$. The final running time for Kutz's algorithm with the modification by Italiano \etal\ is therefore $g^{O(g)} n \log \log n$.

\begin{figure*}[t]
  \centering
  \includegraphics[height=1.5in]{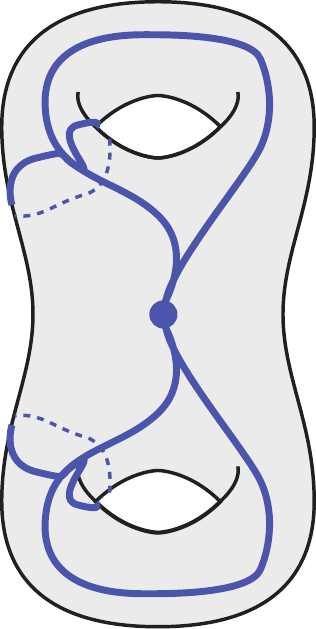}
  \qquad\qquad
  \includegraphics[height=1.5in]{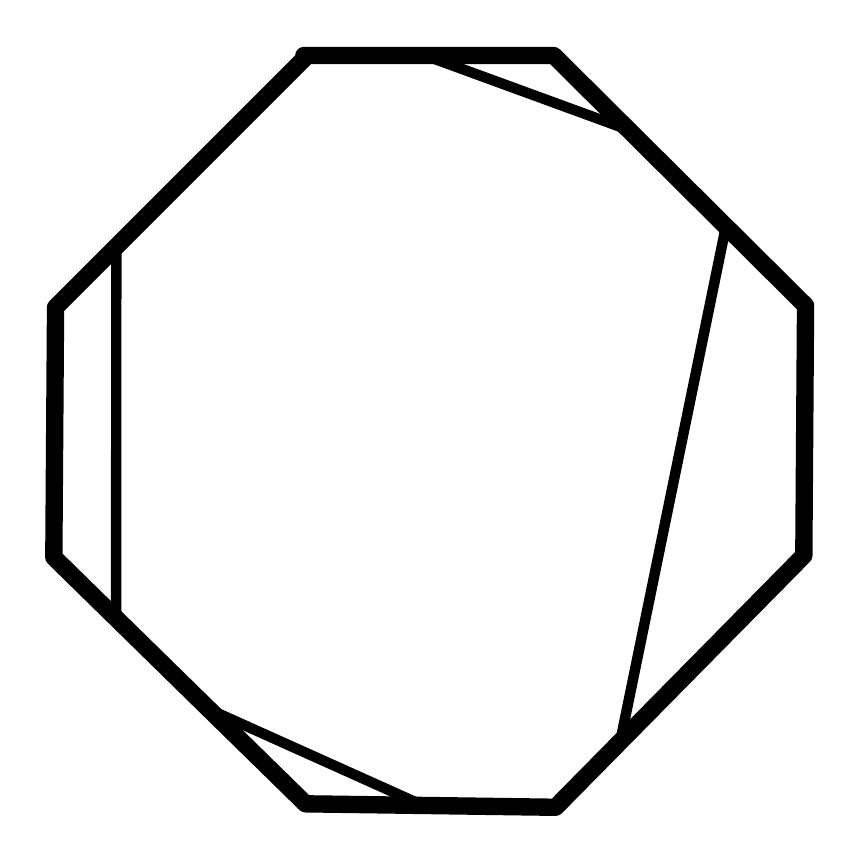}
  \qquad\quad
  \includegraphics[height=1.5in]{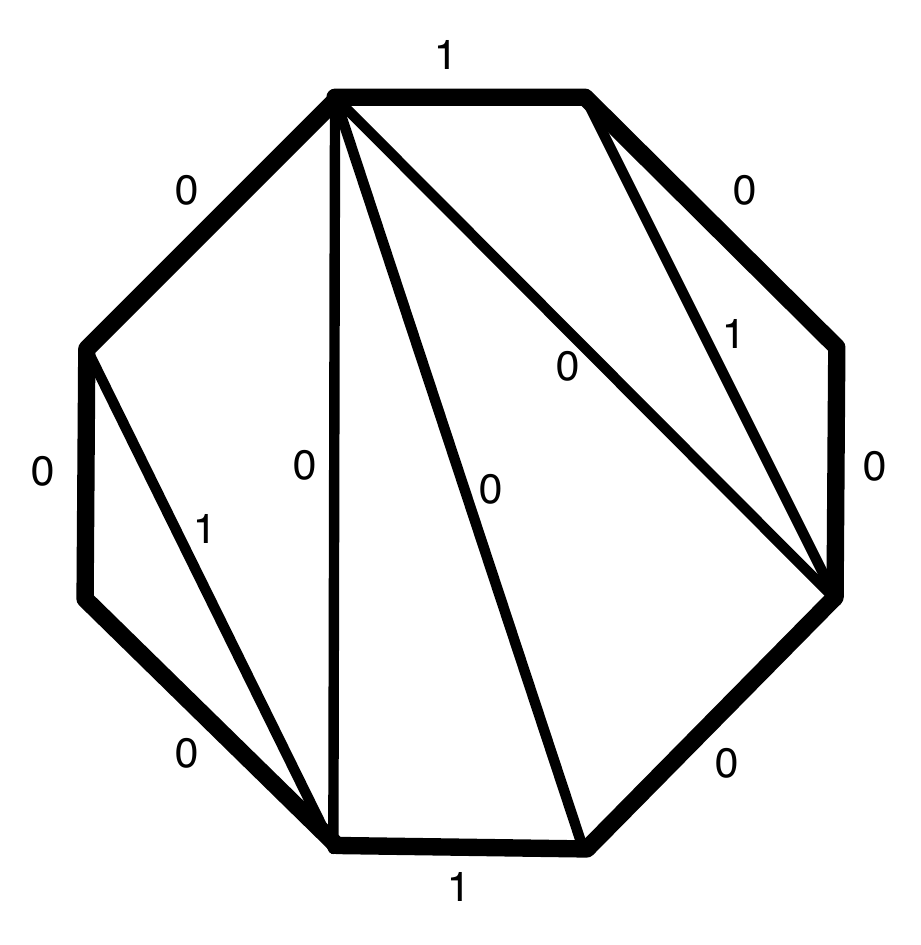}
  \caption{Left: A system of $4$ loops $\Lambda$ on $\Sigma$. Center: Arcs crossing a polygonal schema.
  Right: The weighted triangulation of the dualized schema.}
  \label{fig:triangulations}
\end{figure*}

In order to improve the running time, we show how to reduce the number of crossing sequences that need to be considered by Kutz's algorithm using a similar strategy to that seen in~\cite{ccelw-scsih-08,cen-mcshc-09}.
As mentioned, the greedy system of loops $\Lambda$ used by Kutz cuts the surface into a topological disk~$D$. By replacing each loop in $\Lambda$ with a single edge in~$D$, we transform $D$ into an \emph{abstract polygonal schema}. Each loop of $\Lambda$ corresponds to two edges of the polygon. 
Any non-self-crossing cycle $\gamma$ in $\Sigma$ is cut into arcs by the polygon where an arc exists between two edges if $\gamma$ consecutively crosses the corresponding loops of~$\Lambda$.
We dualize the polygonal schema by replacing each edge with a vertex and each vertex with an edge. Cycle $\gamma$ now corresponds to a \emph{weighted triangulation} of the dualized polygonal schema where each pair of consecutive crossings by $\gamma$ between loops of $\Lambda$ is represented by an edge between the corresponding vertices. Each edge of the triangulation receives a weight equal to the number of times $\gamma$ performs the corresponding consecutive crossings. Some shortest non-trivial cycle crosses each member of $\Lambda$ at most twice, so the edge weights on its triangulation are all between $0$ and $2$.

Our algorithm for computing a shortest non-trivial cycle in $G$ enumerates all weighted triangulations of the dualized polygonal schema with weights between $0$ and $2$ by brute force. There are $2^{O(g)}$ weighted triangulations considered. For each triangulation, the algorithm then checks if it corresponds to a single cycle in $O(g)$ time by brute force. If the triangulation does correspond to a single cycle, then its crossing sequence is calculated. The algorithm uses Italiano \etal's~\cite{insw-iamcmf-11} modification to Kutz's~\cite{k-csnco-06} algorithm to determine if the crossing sequence meets the aforementioned criteria and, if so, to calculate a shortest cycle corresponding to that crossing sequence. Our algorithm will eventually return a shortest cycle corresponding to the correct crossing sequence for some shortest non-trivial cycle. The overall running time is $2^{O(g)} n \log \log n$.

\subsection{Surfaces with Boundary}
We now extend the above algorithm to work on surfaces with boundary.
For computing a shortest non-separating cycle, we reduce to the case without boundary by pasting disks into each of the boundary components. This transformation does not change the set of non-separating cycles. Our algorithm still runs in time $2^{O(g)}n \log \log n$.

In order to compute a shortest non-contractible cycle, we use a \emph{greedy system of $O(g + b)$ arcs}~\cite{ccelw-scsih-08,c-scgsp-10,ce-tnpcs-10,en-mcsnc-11,e-sncds-11} instead of a greedy system of loops. The necessary properties of the greedy system (the shortest non-contractible cycle being simple, crossing each arc at most twice, and being curl-free) are easily shown using the same proofs given in~\cite{k-csnco-06}. We still use a dualized polygonal schema, except it now has $O(g + b)$ vertices, and our algorithm must enumerate $2^{O(g + b)}$ weighted triangulations. The rest of the details are essentially the same. The overall running time is $2^{O(g + b)} n \log \log n$.

Finally, we can compute a shortest non-null-homologous cycle by slightly modifying the algorithm for finding a shortest non-contractible cycle. The only difference is we ignore weighted triangulations where the corresponding crossing sequence does not correspond to a non-null-homologous cycle. Testing if a crossing sequence corresponds to a non-null-homologous cycle can be done using techniques shown in~\cite{cen-mcshc-09}.

With these extensions to surfaces with boundary, we get the following theorem.

\begin{theorem}
A shortest non-separating cycle in an undirected graph embedded on an orientable surface of genus~$g$ with~$b$ boundary cycles can be computed in $2^{O(g)}n \log \log n$ time. Further, a shortest non-contractible or non-null-homologous cycle can be computed in $2^{O(g+b)}n \log \log n$ time.
\end{theorem}


\section{Shortest Non-null-homologous Cycles in Directed Graphs}
\label{sec:non-null-homologous}
Now let~$G$ be a symmetric directed graph with non-negative edge weights, cellularly
embedded on an orientable surface~$\Sigma$ of genus~$g$ with~$b$ boundary
cycles. 
We continue by giving an overview of an algorithm to compute a shortest cycle
in~$G$ that is not null-homologous.
 
In~\cite{e-sncds-11}, Erickson describes a system $\Lambda=\Set{\lambda_1,\lambda_2,\dots,\lambda_{2g}}$ of $2g$ non-separating cycles where each cycle
$\lambda_i$ is composed of two shortest paths in $G$ along with an extra edge. We actually describe and use this construction explicitly in Section~\ref{sec:lifting}. For each cycle $\lambda_i \in \Lambda$,
Erickson gives an $O(g n \log n)$ time algorithm to find a shortest cycle that crosses $\lambda_i$ an odd number of times. Any non-separating cycle must cross at least one member of $\Lambda$ an odd number of times, so an $O(g^2 n \log n)$ time algorithm for finding a shortest non-separating cycle follows immediately.

In a similar vain, we claim it is possible to compute in $O(g n \log n)$ time a shortest cycle crossing any non-separating \emph{arc} $\lambda$ an odd number of times assuming $\lambda$ is a shortest path.
Our algorithm for finding a shortest non-null-homologous cycle begins by calling Erickson's algorithm as a subroutine in case any shortest non-null-homologous cycles are non-separating.
We then perform the following steps
in case all the shortest non-null-homologous cycles are separating.
Arbitrarily label the boundary
cycles of~$G$ as~$B_0,B_1,\dots,B_{b-1}$.
Let~$s$ be an arbitrary vertex on~$B_0$.
We compute the shortest path tree~$T$ from~$s$ using Dijkstra's algorithm
in~$O(n \log n)$ time.
For each index~$i \geq 1$, let~$\lambda_i$ be a shortest directed path in~$T$
from~$B_0$ to~$B_i$ that contains exactly one vertex from each
boundary cycle~$B_0$ and~$B_i$.
Let $\Lambda = \Set{\lambda_1,\lambda_2,\dots,\lambda_{b-1}}$ be the set
of shortest paths computed above. Each path must be non-separating as it
connects two distinct boundary cycles. We can easily compute~$\Lambda$
in~$O(bn)$ time once we have the shortest path tree~$T$.
If a shortest non-null-homologous cycle is separating, then it must
separate~$B_0$ from some other boundary cycle~$B_i$ with~$i\geq1$.

\begin{lemma}
  \label{lem:odd-crossings}
  If a simple cycle~$\gamma$ separates boundary cycle~$B_0$ from a different
  boundary cycle~$B_i$, then~$\lambda_i$ crosses~$\gamma$ an odd number of
  times.
\end{lemma}
\begin{proof}
  Cycle~$\gamma$ separates~$\Sigma$ into two components~$A$ and~$B$ containing
  boundary cycles~$B_0$ and~$B_i$ respectively.
  Arc~$\lambda$ must cross~$\gamma$ from~$A$ to~$B$ one more time than
  it crosses from~$B$ to~$A$.
  Therefore,~$\lambda$ crosses an odd number of times.
\end{proof}
Lemma~\ref{lem:odd-crossings} implies that any shortest
non-null-homologous cycle~$\gamma$ crosses some arc~$\lambda_i$ an odd number of
times if~$\gamma$ is separating.

All that remains is to
present a slightly modified lemma of Erickson~\cite[Lemma 3.4]{e-sncds-11}.
\begin{lemma}
  \label{lem:use-cyclic}
  Let~$\lambda$ be any arc in~$\Lambda$.
  The shortest cycle~$\gamma$ that crosses~$\lambda$ an odd number of times
  can be computed in~$O(g n \log n)$ time.
\end{lemma}
The proof remains essentially unchanged for our version of the lemma.
In short, we compute the cyclic double cover~$\Sigma_\lambda^2$ as described in Appendix~\ref{app:double-cover}.
The lift of~$G = (V,E)$ to~$\Sigma_\lambda^2$ contains the vertex set~$V \times \Set{0,1}$.
Lemma~\ref{lem:lift-to-double-cover} implies that~$\gamma$ lifts to a shortest
path from~$(s,0)$ to~$(s,1)$, for some vertex~$s$ of~$\lambda$.
We can compute this path using a single
multiple-source shortest path computation
in~$O(g n \log n)$ time (Lemma~\ref{lem:mssp}).

Applying Lemma~\ref{lem:use-cyclic} to each arc~$\lambda \in \Lambda$
and comparing the results to the shortest non-separating cycle found
by Erickson's algorithm, we immediately get Theorem~\ref{thm:non-null-homologous}.

\begin{theorem}
\label{thm:non-null-homologous}
  A shortest non-null-homologous cycle in a directed graph embedded on an
  orientable surface of genus~$g$ with~$b$ boundary cycles can be computed
  in~$O((g^2+gb) n \log n)$ time.
\end{theorem}

\section{The Infinite Cyclic Cover}
\label{sec:cover}

As in the previous section, let~$G$ be a symmetric directed graph with
non-negative edge weights, cellularly
embedded on an orientable surface~$\Sigma$ of genus~$g$ with~$b$ boundary
cycles.
We begin to describe
our algorithm for computing a shortest non-contractible cycle in~$G$.
Our job is easy if any shortest non-contractible cycle is non-null-homologous;
we can just run the
algorithm given in Section~\ref{sec:non-null-homologous}
in $O((g^2+gb) n \log n)$ time.
We must work harder, though, to find a shortest
non-contractible cycle~$\gamma$ if
every shortest non-contractible cycle is
null-homologous.
Our high-level strategy is to construct~$O(g)$
subsets of a
covering space we call the infinite cyclic cover.
In Lemma~\ref{lem:non-null-homologous-lift},
we show at least one of the subsets contains a non-null-homologous
cycle that projects to~$\gamma$.

Let~$\lambda$ be an arbitrary simple non-separating cycle in~$\Sigma$.
We define the covering space~$\Sigma_\lambda$, which we call the
\EMPH{infinite cyclic cover}\footnote{Named for the infinite cyclic group.},
as follows. Cutting the surface~$\Sigma$ along~$\lambda$ gives us a new
surface~$\Sigma'$ with~$b+2$ boundary cycles where two of the boundary cycles
are copies of~$\lambda$ denoted~$\lambda^+$ and~$\lambda^-$. The
infinite cyclic cover is
obtained by pasting together an infinite number of copies of~$\Sigma'$ along
corresponding boundary cycles~$\lambda^\pm$. Specifically, we have a
copy~$(\Sigma', i)$ of~$\Sigma'$ for each integer~$i$.
Let~$(\lambda^+,i)$ and~$(\lambda^-,i)$
denote copies of~$\lambda^+$ and~$\lambda^-$ in~$(\Sigma',i)$.
The infinite cyclic cover
is defined by identifying~$(\lambda^+,i)$ and~$(\lambda^-,i+1)$
for every~$i$. Any graph~$G$ cellularly embedded on~$\Sigma$ lifts to an
infinite graph~$G_\lambda$ embedded in~$\Sigma_\lambda$. Note that for
any pair of simple non-separating cycles~$\lambda$ and~$\mu$, the infinite
cyclic covers~$\Sigma_\lambda$ and~$\Sigma_\mu$ are homeomorphic, but the
lifted graphs~$G_\lambda$ and~$G_\mu$ may not be isomorphic.

\begin{figure*}[t]
  \centering
  \includegraphics[height=1.5in]{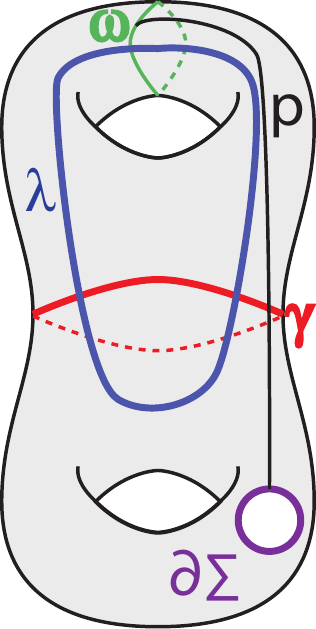}\\
  \vspace{0.3in}
  \includegraphics[height=1.5in]{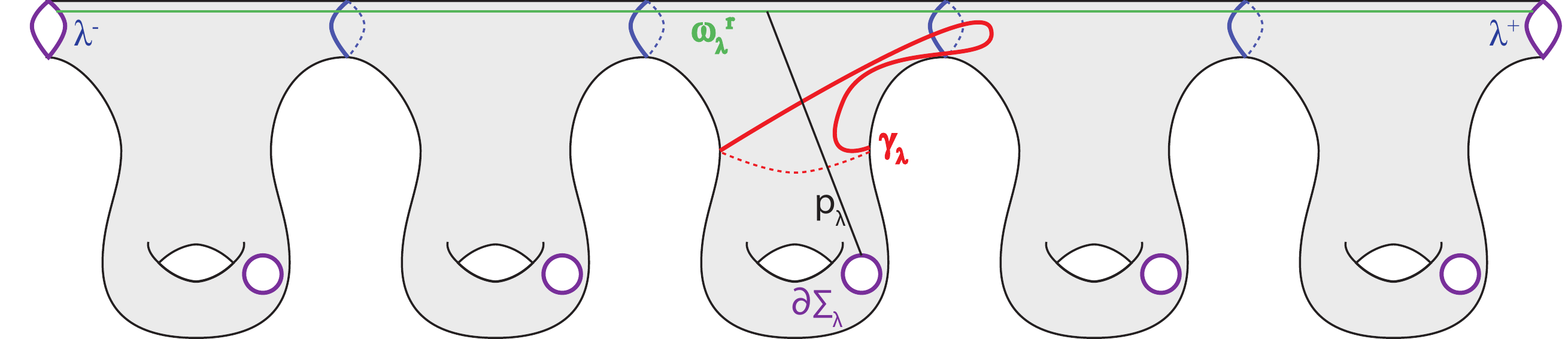}
  \caption{The restricted infinite cyclic cover and proof of Lemma~\ref{lem:non-null-homologous-lift}.
  Top: The surface~$\Sigma$ with boundary~$\partial \Sigma$. Bottom: The surface~$\Sigma_\lambda^r$;
  cycle~$\gamma_\lambda$ separates~$\partial \Sigma_\lambda$ from~$\lambda^-$.}
  \label{fig:restricted}
\end{figure*}

We would like to use the infinite cyclic cover to aid us in finding a shortest
non-contractible cycle.
As explained in Section~\ref{sec:lifting},
it is possible to consider only a finite
portion of~$\Sigma_\lambda$ if we choose $\lambda$ carefully.
We call this subset the \EMPH{restricted
infinite cyclic cover}.
Again, let~$\lambda$ be an arbitrary simple non-separating cycle in~$\Sigma$
and define~$\Sigma'$ as above with boundaries~$\lambda^+$ and~$\lambda^-$.
Instead of pasting together an infinite number
of copies of~$\Sigma'$, we only paste together five copies. Specifically,
we have a copy~$(\Sigma',i)$ of~$\Sigma'$ for each
integer~$i \in \Set{1,\dots,5}$.
Again, let~$(\lambda^+,i)$ and~$(\lambda^-,i)$
denote copies of~$\lambda^+$ and~$\lambda^-$ in~$(\Sigma',i)$.
The restricted infinite cyclic cover is defined by identifying~$(\lambda^+,i)$
and~$(\lambda^-,i+1)$ for every~$i \in \Set{1,\dots,4}$.
See Figure~\ref{fig:restricted}.
Now any graph~$G$
cellularly embedded on~$\Sigma$ lifts to a \emph{finite} graph~$G_\lambda^r$
embedded in~$\Sigma_\lambda^r$ with at most six times as many
vertices and edges.
Note that~$\Sigma_\lambda^r$ still has two lifts of~$\lambda$ acting as boundary
cycles.
We continue to refer to these boundary cycles
as~$\lambda^+$ and~$\lambda^-$ when it
is clear from context that we are referring to the restricted infinite
cyclic cover.
Euler's formula implies the genus of $\Sigma_\lambda^r$ is $5g - 5$.


Further restrict~$\lambda$ to be a simple non-separating cycle in~$G$.
For any path or cycle~$p$,
we define the
\EMPH{crossing count} $c_\lambda(p)$ to be the number of times~$p$
crosses~$\lambda$ from left to right \emph{minus} the number of times~$p$
crosses~$\lambda$ from right to left. Equivalently, we have
$$c_\lambda(p) = \sum_{u \arcto v \in p} c_\lambda(u \arcto v)$$
where for any directed edge $u \arcto v$, we define $c_\lambda(u \arcto v)$
to be~$1$
if $u \arcto v$ \emph{enters} $\lambda$ from the left,~$-1$ if $u \arcto v$
\emph{leaves} $\lambda$ from the left,
and~$0$ otherwise.
We can define the restricted infinite cyclic cover using a \emph{voltage construction}~\cite[Chapters 2,4]{gt-tgt-01}
for combinatorial surfaces.
Let~$G_\lambda^r$ be the graph whose vertices are the pairs~$(v,i)$,
where~$v$ is a vertex of~$G$ and~$i$ is an integer in~$\Set{1,\dots,6}$
if~$v$ lies along~$\lambda$ or~$\Set{1,\dots,5}$ if~$v$ does not lie
along~$\lambda$.
The edges of~$G_\lambda^r$ are the ordered pairs
$$(u \arcto v,i) := (u,i) \arcto (v,i+c_\lambda(u\arcto v))$$
for all edges~$u \arcto v$ of~$G$ and all~$i \in \Set{1,\dots,6}$.
Let~$\pi : G_\lambda^r \to G$ denote the obvious covering map~$\pi(v,i) = v$.
We declare that a cycle in~$G_\lambda^r$ bounds a face of~$G_\lambda^r$ if and
only if its projection to~$G$ bounds a face of~$G$. The resulting embedding
of~$G_\lambda^r$ defines the restricted infinite cyclic
cover~$\Sigma_\lambda^r$.

\section{Lifting Shortest Non-contractible Cycles}
\label{sec:lifting}

Consider the following procedure also used in~\cite{e-sncds-11}.
We construct a
greedy tree-cotree decomposition~$(T,L,C)$ of~$G$, where~$T$ is a shortest
path tree rooted at some arbitrary vertex of~$G$. Euler's formula implies
that~$L$ contains exactly~$2g$ edges; label these edges arbitrarily
as~$u_1v_1,u_2v_2,\dots,u_{2g}v_{2g}$. For each index~$i$, let~$\lambda_i$
denote the unique cycle in the \emph{undirected} graph~$T \cup u_i v_i$
oriented so that is contains the directed edge~$u_i \arcto v_i$. If
there are no boundary cycles in~$\Sigma$, then the set of
cycles~$\Lambda=\Set{\lambda_1,\lambda_2,\dots,\lambda_{2g}}$ is a basis
for the first homology group of~$\Sigma$~\cite{e-dgteg-03}. We
refer to the construction as a \EMPH{partial homology basis}.
Every non-separating cycle in~$\Sigma$ crosses at least one cycle
in~$\Lambda$ an odd number of times~\cite[Lemma 3]{cm-fsnsn-07}. The
greedy tree-cotree decomposition~$(T,L,C)$ can be constructed in~$O(n \log n)$
time using Dijkstra's algorithm.
Afterward, we can easily compute the partial homology basis in~$O(gn)$ time.

Recall that a single period lift of a cycle~$\gamma$ to a covering space refers to any lift
of a loop equivalent to~$\gamma$. Let~$\sigma$ be an arbitrary shortest path in~$G$.
Erickson~\cite{e-sncds-11} argues that the lift of any shortest
non-contractible cycle to the universal cover does not intersect many lifts of~$\sigma$.
This observation applies to the infinite cyclic cover as well.
The following lemma and its corollary
are essentially equivalent to Lemma 4.6 and Corollary 4.7 of~\cite{e-sncds-11}, but
modified for our setting.

\begin{lemma}
  \label{lem:intersect-paths}
  Let~$\gamma$ be a shortest non-contractible cycle in~$\Sigma$;
  let~$\lambda$ be any simple non-separating cycle in~$\Sigma$;
  and
  let~$\sigma$ be any shortest path in~$\Sigma$. Any single-period lift
  of~$\gamma$ to
  the infinite cyclic
  cover~$\Sigma_{\lambda}$ intersects at most two lifts of~$\sigma$.
\end{lemma}
\begin{proof}
  The covering space~$\Sigma_\lambda$ is path connected, so it is
  itself covered by the \emph{universal cover}~$\tilde{\Sigma}$.
  Any single period lift of~$\gamma$ to~$\Sigma_\lambda$ in turn has one or
  more lifts in~$\tilde{\Sigma}$.
  Any one of these single period lifts of~$\gamma$ to~$\tilde{\Sigma}$
  intersects at most two lifts of~$\sigma$~\cite[Lemma 4.6]{e-sncds-11}.
  Covering maps are functions, so
  lifting from~$\Sigma_\lambda$ to~$\tilde{\Sigma}$ cannot decrease
  the number of intersecting lifts of~$\sigma$.
\end{proof}
\begin{corollary}
  \label{cor:crossing-limit}
  Let~$\Lambda$ be a partial homology basis in~$\Sigma$;
  let~$\lambda$ be any \emph{cycle} in~$\Lambda$; and
  let~$\gamma$ be a shortest non-contractible cycle in~$\Sigma$.
  Any single-period lift of~$\gamma$ to~$\Sigma_\lambda$ intersects at most
  four lifts of~$\lambda$.
\end{corollary}
\begin{proof}
  Every vertex of~$\lambda$ belongs to one of two directed shortest paths.
  By Lemma~\ref{lem:intersect-paths},
  any single-period lift of~$\gamma$ intersects at most two lifts of either
  shortest path.
\end{proof}

Recall the restricted infinite cyclic cover defined in
Section~\ref{sec:cover} is constructed by pasting together \emph{five}
copies of the surface cut along the simple non-separating cycle~$\lambda$.
We immediately get the following lemma stating the restricted infinite cyclic cover is large
enough to contain a lift of any shortest non-contractible cycle.

\begin{lemma}
  \label{lem:fits}
  Let~$\Lambda$ be a partial homology basis in~$\Sigma$;
  let~$\lambda$ be any cycle in~$\Lambda$; and
  let~$\gamma$ be a shortest non-contractible cycle in~$\Sigma$.
  There exists a single period lift of~$\gamma$ to~$\Sigma_\lambda^r$.
\end{lemma}

In fact,
we show below that $\gamma$ lifts to be a \emph{shortest non-contractible cycle} in $\Sigma_\lambda^r$ if~$\gamma$ is separating.
This statement actually holds for \emph{any} non-separating cycle $\lambda$ made of two shortest paths optionally connected by an edge.
In Lemma~\ref{lem:non-null-homologous-lift}, we explain that the correct
choice of~$\lambda$ guarantees the lift of $\gamma$ to be non-null-homologous.

We continue by noting that every shortest non-contractible cycle is
simple~\cite[Lemma 3]{ccl-fsncd-10}.
We show that if any shortest non-contractible cycle~$\gamma$ is
\emph{separating}, then it lifts to a cycle
in~$\Sigma_\lambda^r$ for any~$\lambda$ in the partial homology basis.
Recall the definition of the crossing count~$c_\lambda(\gamma)$.

\begin{lemma}
  \label{lem:lift-as-cycle}
  Let~$\lambda$ be any simple non-separating cycle in~$\Sigma$,
  and let~$\gamma$ be a loop in~$\Sigma$ with a lift in~$\Sigma_\lambda^r$.
  Then,~$\gamma$ lifts to a loop in~$\Sigma_\lambda^r$ if and only if
  $c_\lambda(\gamma) = 0$.
\end{lemma}
\begin{proof}
  Let~$\Sigma'$ be the surface~$\Sigma$ cut along~$\lambda$.
  By construction,~$\Sigma_\lambda^r$ is composed of five copies
  of~$\Sigma'$, denoted~$(\Sigma', i)$ for each integer~$i \in \Set{1,\dots,5}$.
  Each copy is separated by a lift of~$\lambda$.
  Consider a lift of~$\gamma$ 
  contained in~$\Sigma_\lambda^r$ which we denote~$\gamma_\lambda$.
  For every instance of~$\gamma$ crossing~$\lambda$ from left to right,
  there is an instance of~$\gamma_\lambda$ crossing a lift of~$\lambda$
  from~$(\Sigma',i)$ to~$(\Sigma',i+1)$ for some~$i$.
  Likewise, every time~$\gamma$ crosses~$\lambda$ from right to left,
  $\gamma_\lambda$ crosses~$\lambda$ from~$(\Sigma',i)$ to~$(\Sigma',i-1)$.
  If~$\gamma_\lambda$ begins in~$(\Sigma',i)$, then it
  ends at a copy of the same point in~$(\Sigma',i+c_\lambda(\gamma))$.
\end{proof}
\begin{lemma}
  \label{lem:crossing-count}
  Let~$\lambda$ be any simple non-separating cycle in~$\Sigma$
  and let~$\gamma$ be any simple separating cycle. We
  have~$c_\lambda(\gamma) = 0$.
\end{lemma}
\begin{proof}
  Cycle~$\gamma$ separates~$\Sigma$ into two components denoted~$A$ and~$B$
  so that a path crossing~$\gamma$ exactly once starts in~$A$ and ends in~$B$
  if it crosses from left to right.
  Let~$x$ be an arbitrary point on~$\lambda$ and consider the loop~$\ell$
  equivalent to~$\lambda$ based at~$x$.
  Every time~$\ell$
  crosses~$\gamma$ from left to right, we see~$\ell$ goes from~$A$ to~$B$.
  Further~$\gamma$ crosses~$\ell$ once from right to left. Similarly,
  every time~$\ell$ crosses~$\gamma$ from right to left, we see~$\ell$ goes
  from~$B$ to~$A$ and~$\gamma$ crosses~$\ell$ once from left to right.
  Loop~$\ell$ must cross from~$A$ to~$B$ the same number of times
  it crosses from~$B$ to~$A$. Therefore,~$\gamma$ crosses~$\ell$ and~$\lambda$
  from right to left the same number of times it crosses left to right.
  By definition,~$c_\lambda(\gamma) = 0$.
\end{proof}

\begin{corollary}
  Let~$\Lambda$ be a partial homology basis in~$\Sigma$;
  let~$\lambda$ be any cycle in~$\Lambda$; and
  let~$\gamma$ be a shortest non-contractible cycle in~$\Sigma$.
  If~$\gamma$ is separating, then~$\gamma$ lifts to \emph{a loop}
  in~$\Sigma_\lambda^r$.
\end{corollary}

We can finally show that if any shortest non-contractible cycle~$\gamma$ is
separating, then it actually lifts to a shortest non-contractible
cycle in~$\Sigma_\lambda^r$ for any~$\lambda$ in a partial homology basis.

\begin{lemma}
  \label{lem:contractible}
  Let~$\gamma_\lambda$ be a loop in~$\Sigma_\lambda^r$ that projects
  to a simple loop~$\gamma$ in~$\Sigma$. Loop~$\gamma_\lambda^r$
  is contractible if and only if~$\gamma$ is contractible.
\end{lemma}
\begin{proof}
  Suppose~$\gamma_\lambda$ is contractible. There exists a homotopy~$h$
  from~$\gamma_\lambda$ to a constant map. The paths in~$h$ can be projected
  to~$\Sigma$, yielding a homotopy from~$\gamma$ to a constant map.
  Therefore,~$\gamma$ is contractible.

  Now, suppose~$\gamma$ is contractible. There exists a homotopy~$h$
  from~$\gamma$ to a constant map. There exists a unique homotopy~$h_\lambda$
  of~$\gamma_\lambda$ that lifts the paths in~$h$ to
  the infinite cyclic cover~$\Sigma_\lambda$~\cite[Proposition 1.30]{h-at-02}.
  Homotopy~$h_\lambda$ finishes with a constant map, so~$\gamma_\lambda$
  is contractible in~$\Sigma_\lambda$. Loop~$\gamma_\lambda$ must be
  simple to project to a simple loop~$\gamma$, so it bounds a disk~$D$
  in~$\Sigma_\lambda$. Disk~$D$ contains no faces outside of~$\Sigma_\lambda^r$,
  because~$\gamma_\lambda$ contains no edges outside of~$\Sigma_\lambda^r$ to
  bound those outside faces. Therefore,~$\gamma_\lambda$ bounds a disk ($D$)
  in~$\Sigma_\lambda^r$ implying~$\gamma_\lambda$ is contractible
  in~$\Sigma_\lambda^r$.
\end{proof}

\begin{lemma}
  \label{lem:shortest-non-contractible-lift}
  Let~$\Lambda$ be a partial homology basis in~$\Sigma$;
  let~$\lambda$ be any cycle in~$\Lambda$; and
  let~$\gamma$ be a shortest non-contractible cycle in~$\Sigma$.
  If~$\gamma$ is separating,
  then~$\gamma$ lifts to a shortest non-contractible
  cycle in~$\Sigma_\lambda^r$.
\end{lemma}

\section{Computing Shortest Non-contractible Cycles in Directed Graphs}
\label{sec:non-contractible-algorithm}

We now
describe our algorithm for computing a shortest non-contractible cycle.
We assume the surface has genus~$g \geq 1$. Otherwise, every non-contractible
cycle is non-null-homologous, and we can simply use the algorithm
given in Section~\ref{sec:non-null-homologous}.
Further,
we begin by assuming the surface has exactly one boundary cycle. Instances
where~$\Sigma$ has more than one boundary cycle or no boundary cycles
are handled as simple reductions to the one boundary cycle case given at the end of this section.

Let~$\partial \Sigma$ denote the one boundary cycle on~$\Sigma$.
We compute a partial homology basis
${\Lambda = \Set{\lambda_1, \lambda_2, \dots, \lambda_{2g}}}$
in~$O(n \log n + gn)$ time as described in Section~\ref{sec:lifting}.
The following lemma states that one of the cycles in the homology basis
can be used to build a restricted infinite cyclic cover that is useful
for our computation. Surprisingly, the boundary introduced by restricting
the infinite cyclic cover plays a key role in the proof of the lemma.

\begin{lemma}
  \label{lem:non-null-homologous-lift}
  Let~$\gamma$ be a shortest non-contractible cycle in~$\Sigma$.
  If~$\gamma$ is separating, then
  there exists a non-separating cycle~$\lambda \in \Lambda$ such
  that~$\gamma$ lifts to a shortest non-null-homologous cycle
  in the restricted infinite cyclic cover~$\Sigma_\lambda^r$.
\end{lemma}
\begin{proof}
 Every shortest non-contractible cycle is simple~\cite[Lemma 3]{ccl-fsncd-10}.
  So by assumption,~$\gamma$ is a simple separating cycle.
  There is exactly one boundary~$\partial \Sigma$,
  so~$\gamma$ bounds the closure~$A$ of a set of faces.
  The component~$A$ must have genus, or~$\gamma$ would bound a disk
  and be contractible.
  There exists a simple non-separating cycle~$\omega$ on~$\Sigma$ contained
  entirely within~$A$. Cycle~$\omega$ must cross some
  other cycle~$\lambda \in \Lambda$ an odd number of
  times~\cite[Lemma 3]{cm-fsnsn-07}.
  See Figure~\ref{fig:restricted}.
  Consider the infinite cyclic cover~$\Sigma_\lambda$ and its
  restriction~$\Sigma_\lambda^r$.

  Let~$p$ be a path in~$\Sigma$ from~$\partial \Sigma$ to~$\omega$ such
  that~$p$ does not cross~$\lambda$. Path~$p$ must exist, because~$\lambda$
  is non-separating. Further,~$p$ crosses~$\gamma$ an odd number of times.
  Let~$\partial \Sigma_\lambda$ be a lift
  of~$\partial \Sigma$ to~$\Sigma_\lambda$,
  and let~$p_\lambda$ be the lift of~$p$ to~$\Sigma_\lambda$ that begins
  on~$\partial \Sigma_\lambda$.
  Let~$\gamma_\lambda$ be a lift of~$\gamma$ to~$\Sigma_\lambda$ such
  that~$p_\lambda$ crosses~$\gamma_\lambda$ an odd number of times.
  By symmetry and Lemma~\ref{lem:shortest-non-contractible-lift}, we may
  assume~$\gamma_\lambda$ is a cycle in~$\Sigma_\lambda^r$.
  We note~$\gamma_\lambda$ is simple as it projects to simple cycle~$\gamma$.

  Suppose that~$\gamma_\lambda$ is separating.
  Let~$\omega_\lambda$ denote a lift of cycle~$\omega$ to~$\Sigma_\lambda$
  such that~$p_\lambda$ ends on~$\omega_\lambda$.
  Curve~$\omega_\lambda$ is not a cycle in~$\Sigma_\lambda$, because~$\omega$
  crosses~$\lambda$ an odd number of times in~$\Sigma$
  (see Lemma~\ref{lem:lift-as-cycle}).
  Therefore,~$\omega_\lambda$ is
  a simple infinite path
  that does not cross any lift of~$\gamma$.
  Let~$\omega_\lambda^r = \omega_\lambda \cap \Sigma_\lambda^r$.
  Path~$\omega_\lambda^r$ is a simple arc from~$\lambda^-$ to~$\lambda^+$
  in~$\Sigma_\lambda^r$ which does not cross~$\gamma_\lambda$.
  Path~$p$ does not
  cross~$\lambda$, implying that~$p_\lambda$ is a path
  in~$\Sigma_\lambda^r$ with endpoints on~$\partial \Sigma_\lambda$
  and~$\omega_\lambda^r$. Further,~$p_\lambda$ crosses~$\gamma_\lambda$
  an odd number of times, implying that~$\gamma_\lambda$
  separates~$\partial \Sigma_\lambda$ from~$\omega_\lambda^r$ and~$\lambda^-$.

  We see either~$\gamma_\lambda$ is non-separating or it separates a pair
  of boundary cycles.
  Therefore,~$\gamma_\lambda$ is non-null-homologous
  in~$\Sigma_\lambda^r$ .
  Lemma~\ref{lem:shortest-non-contractible-lift} implies~$\gamma_\lambda$
  is actually a shortest non-null-homologous cycle in~$\Sigma_\lambda^r$.
\end{proof}

In the above proof, it would actually be preferable
if~$\gamma_\lambda$ \emph{was} separating. In this case,
we could find~$\gamma_\lambda$ in~$O(g n \log n)$ time by
applying Lemma~\ref{lem:use-cyclic} along shortest paths
between~$\lambda^-$ and each lift of~$\partial \Sigma$.
As written, the lemma requires
us to apply the full algorithm of Section~\ref{sec:non-null-homologous}
in~$O(g^2 n \log n)$ time if we wish to find~$\gamma_\lambda^r$.

We now finish considering the case where~$\Sigma$ has one boundary cycle.
Applying lemmas~\ref{lem:shortest-non-contractible-lift} and~\ref{lem:non-null-homologous-lift},
we construct the restricted infinite cyclic cover~$\Sigma_\lambda^r$ and
find a shortest non-null-homologous cycle in~$\Sigma_\lambda^r$ once for
each cycle~$\lambda \in \Lambda$ using
the algorithm of Section~\ref{sec:non-null-homologous}.
This procedure gives us a
shortest non-contractible cycle in~$O(g^3 n \log n)$ time if any are separating.
We apply the algorithm
of Section~\ref{sec:non-null-homologous} (or Erickson's~\cite{e-sncds-11}
algorithm) once to~$G$ directly to account for the case where every
shortest non-contractible cycle is non-separating. All that remains is to consider the
cases where~$\Sigma$ has several boundary cycles or no boundary cycles.

\subsection{Surfaces with Several Boundary}

We now consider the case where~$\Sigma$ has~$b > 1$ boundary cycles.
We apply the algorithm of Section~\ref{sec:non-null-homologous} to find
any shortest non-contractible cycles that are non-null-homologous.
Next, we paste disks into all but one of the boundary cycles.
This transformation does not introduce any new non-contractible cycles,
because it does not remove any paths from any homotopies.
Further, it does not restrict the set of
non-contractible null-homologous cycles. Every such cycle~$\gamma$ still
separates a subset of faces (with genus) from the one remaining boundary
cycle. We now apply the algorithm as given for one boundary cycle to
find any shortest non-contractible cycles that happen to be
null-homologous.

\subsection{Surfaces without Boundary}

Finally, we extend our algorithm
to consider the case where~$\Sigma$ has no boundary.
We apply the algorithm of Section~\ref{sec:non-null-homologous}
to find any shortest non-contractible cycles that are non-null-homologous
(we can also apply Erickson's~\cite{e-sncds-11}
algorithm as every non-null-homologous cycle is also non-separating
on a surface without boundary).
We then perform the following reduction in case every shortest
non-contractible cycle is null-homologous.
We compute \emph{one} cycle~$\lambda$ of a greedy homology basis
using a greedy tree-cotree decomposition in~$O(n \log n)$ time
and reduce the problem of finding the shortest non-contractible
cycle for the surface~$\Sigma$ with genus~$g$ and
no boundary to
the \emph{same} problem on the larger surface~$\Sigma_\lambda^r$, which
has two boundary cycles and genus~$5g-5$.
Note that the shortest non-contractible cycle in~$\Sigma_\lambda^r$
may be non-separating.
The reduction is correct according to
Lemma~\ref{lem:shortest-non-contractible-lift}. We then apply the algorithm for several
boundary on the new surface~$\Sigma_\lambda^r$.
Using both extensions and the algorithm as given above, we get our desired theorem.

\begin{theorem}
  A shortest non-contractible cycle in a directed graph embedded on an
  orientable surface of genus~$g$ with~$b$ boundary cycles can be computed
  in~$O((g^3 + gb)n \log n)$ time.
\end{theorem}

\section{Conclusions and Future Work}

We gave algorithms to compute shortest non-trivial cycles in both directed and undirected surface embedded graphs.
In undirected graphs, our algorithms find shortest non-contractible and non-null-homologous cycles in $2^{O(g+b)} n \log \log n$ time and shortest non-separating cycles in $2^{O(g)} n \log \log n$ time.
For directed graphs, our algorithms find shortest non-null-homologous cycles in $O((g^2 + gb)n \log n)$ time and shortest non-contractible cycles in ${O((g^3 + gb)n \log n)}$ time.

The most obvious question remaining is whether we can reduce these times further. In particular, it is natural to ask if we can compute a shortest non-contractible cycle in a directed surface graph in $O((g^2 + gb)n \log n)$ time, matching the algorithm of Cabello \etal~\cite{cc-msspg-07,cce-msspe-12} for undirected surface graphs. The main bottleneck appears to be the need to compute shortest non-null-homologous cycles in the restricted infinite cyclic cover. If the proof of Lemma~\ref{lem:non-null-homologous-lift} can be improved to show an appropriate arc or cycle of $\Sigma_\lambda^r$ is crossed an odd number of times by the lift of a shortest non-contractible cycle, then we can easily reduce the cost of searching each cover to $O(g n \log n)$. Another question is whether or not the $O(n \log \log n)$ running time achieved by Italiano \etal~\cite{insw-iamcmf-11} can be achieved in directed graphs and if its use requires lifting to subsets of the universal cover. 

\paragraph{Acknowledgments.} The author would like to thank Erin W. Chambers, Jeff Erickson, Amir Nayyeri for many helpful discussions as well as the anonymous reviewers for their suggestions.

\bibliographystyle{newabuser}
\bibliography{topology,data-structures,optimization}

\part*{Appendix}
\appendix

\section{Extending the Cyclic Double Cover}
\label{app:double-cover}
%
%
\begin{figure*}[t]
  \centering
  \includegraphics[width=2in]{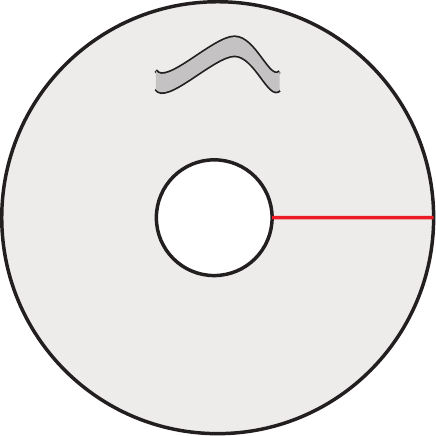} \quad
  \includegraphics[width=2in]{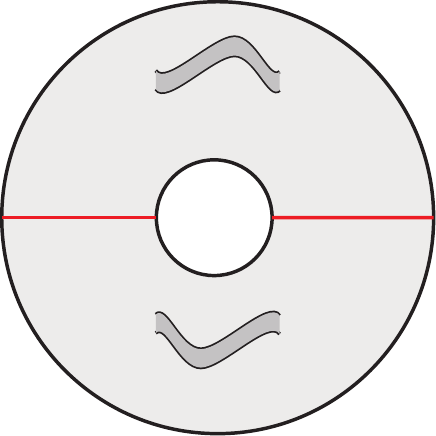}
  \caption{Left: An arc~$\lambda$ between two boundary on a torus.
  Right: The cyclic double cover~$\Sigma_\lambda^2$.}
  \label{fig:double-cover}
\end{figure*}

Let~$G$ be a symmetric directed graph with non-negative edge weights, cellularly
embedded on an orientable surface~$\Sigma$ of genus~$g$ with~$b$ boundary
cycles.
We describe an extension to the cyclic double cover of
Erickson~\cite{e-sncds-11} that works with simple arcs instead of
cycles. Let~$\lambda$ be an arbitrary simple non-separating arc in~$\Sigma$.

Define the covering space~$\Sigma_{\lambda}^2$, which we call the \EMPH{cyclic
double cover}\footnote{Named for the cyclic group of order~$2$.} as follows.
Cutting the surface~$\Sigma$ along~$\lambda$ gives
us a new surface~$\Sigma'$ with at least one boundary cycle.
One boundary cycle
of~$\Sigma'$ contains two copies of~$\lambda$ denoted~$\lambda^+$
and~$\lambda^-$.
Let~$(\Sigma', 0)$ and~$(\Sigma',1)$ denote two distinct copies of~$\Sigma'$.
For any point~$p \in \Sigma'$, let~$(p, 0)$ and~$(p, 1)$ denote the
corresponding points in~$(\Sigma', 0)$ and~$(\Sigma', 1)$, respectively.
In particular, let~$(\lambda^+, 0)$ and~$(\lambda^-, 0)$ denote the copies
of~$\lambda^+$ and~$\lambda^-$ in~$(\Sigma', 0)$.
Finally, let~$\Sigma_\lambda^2$
be the surface obtained by identifying~$(\lambda^+, 0)$ and~$(\lambda^-, 1)$ to
a single arc, denoted~$(\lambda,0)$, and identifying~$(\lambda^+, 1)$
and~$(\lambda^-, 0)$ to a single arc, denoted~$(\lambda, 1)$.
Any graph~$G$
that is cellularly embedded in~$\Sigma$ lifts to a graph~$G_\lambda^2$ with
twice as many vertices and edges that is cellularly embedded
in~$\Sigma_\lambda^2$. There are also twice as many faces in the embedding of $G_\lambda^2$ on
$\Sigma_\lambda^2$ and at least $2b -2$ boundary cycles, so Euler's formula implies the genus of
$\Sigma_\lambda^2$ is at most $2g$.
See Figure~\ref{fig:double-cover}.

For combinatorial surfaces, we can equivalently
define the cyclic double cover using
a standard \emph{voltage construction}~\cite[Chapters 2,4]{gt-tgt-01}.
Here, we assume~$\lambda$ is an arc in~$G$.
For any directed edge~$u \arcto v$, we define~$\varepsilon_\lambda(u \arcto v)$
to be~$1$ if~$u \arcto v$ enters~$\lambda$ from the left or leaves~$\lambda$
from the left,
and~$0$ otherwise.
Let~$G_\lambda^2$ be the graph whose vertices are the pairs~$(v,z)$, where~$v$
is a vertex of~$G$ and~$z$ is a bit, and whose edges are the ordered pairs
$$(u \arcto v, z) := (u, z) \arcto (v, z \oplus \varepsilon_\lambda(u \arcto v))$$
for all edges~$u \arcto v$ of~$G$ and both bits~$z$.
Here,~$\oplus$ denotes addition modulo~$2$.
Let ${\pi : G_\lambda^2 \to G}$
denote the obvious covering map~$\pi(v,z) = v$. We declare that a cycle
in~$G_\lambda^2$ bounds a face of~$G_\lambda^2$ if and only if its
projection to~$G$ bounds a face of~$G$. The resulting embedding of~$G_\lambda^2$
defines the cyclic double cover~$\Sigma_\lambda^2$.
For any directed cycle~$\gamma$, we define the
\EMPH{crossing parity}~$\varepsilon_\lambda(\gamma)$ to be~$1$ if~$\gamma$
crosses~$\lambda$ an odd number of times and~$0$ otherwise. Equivalently,
we have
$$\varepsilon_\lambda(\gamma)
  = \bigoplus_{u \arcto v \in \gamma} \varepsilon_\lambda(u \arcto v).$$

As in~\cite{e-sncds-11}, the following lemmas are immediate.

\begin{lemma}
  Let~$\lambda$ be any simple non-separating arc in~$\Sigma$;
  let~$\gamma$ be any cycle in~$\Sigma$;
  and let~$s$ be any vertex of~$\gamma$. Then~$\gamma$ is the
  projection of a unique path in~$\Sigma_\lambda^2$ from~$(s,0)$
  to~$(s,\varepsilon_\lambda(\gamma))$.
\end{lemma}

\begin{lemma}
  Let~$\lambda$ be any simple non-separating arc in~$\Sigma$. Every lift of
  a shortest directed path in~$G$ is a shortest directed path in~$G_\lambda^2$.
\end{lemma}

\begin{lemma}
  \label{lem:lift-to-double-cover}
  Let~$\lambda$ be any simple non-separating arc in~$\Sigma$;
  let~$\gamma$ be the shortest cycle in~$\Sigma$ that crosses~$\lambda$ an odd
  number of times;
  and let~$s$ be any vertex of~$\gamma$.
  Then~$\gamma$ is the projection of a shortest path in~$\Sigma_\lambda^2$
  from~$(s,0)$ to~$(s,1)$.
\end{lemma}

\end{document}